\documentclass[twocolumn,10pt,a4paper,conference]{IEEEtran}

\usepackage{latexsym}
\usepackage{amssymb}
\usepackage{bm}
\usepackage{stmaryrd}
\usepackage[dvips]{graphics}
\usepackage{graphicx}
\usepackage{psfrag}
\usepackage{amsfonts,amsmath,amssymb,amsthm}
\usepackage{algorithm}
\usepackage{algorithmic}
\usepackage{dsfont,color,subfigure,setspace}
\usepackage{cite,pdfsync}
\usepackage{float}

\newcommand{\X}{X}
\newcommand{\x}{x}

\newcommand{\indx}{i}

\newcommand{\Xbbf}{{\mathbb{\mathbf{X}}}}

\newcommand\blfootnote[1]{%
  \begingroup
  \renewcommand\thefootnote{}\footnote{#1}%
  \addtocounter{footnote}{-1}%
  \endgroup
}

\usepackage{xspace}
\usepackage{bbm}

\newcommand{\Xv}{{\bf X}}

\newcommand{\Xh}{{\hat{X}}}

\newcommand{\xh}{{\hat{x}}}

\def\d{\delta}
\def\e{\epsilon}

\DeclareMathOperator\E{E}

\newcommand\ie{i.e.,\xspace}
\def\textiid{i.i.d.\@\xspace}
\newcommand\iid{\ifmmode\text{ i.i.d. } \else \textiid \fi}

\newtheorem{theorem}{Theorem}
\newtheorem{lemma}{Lemma}

\newtheorem{corollary}{Corollary}
\newtheorem{definition}{Definition}

\theoremstyle{definition}
\newtheorem{remark}{Remark}

\begin{document}
\title{Rate-Distortion Dimension of Stochastic Processes}


\author{
    \IEEEauthorblockN{Farideh Ebrahim Rezagah\IEEEauthorrefmark{1},  Shirin Jalali\IEEEauthorrefmark{2}, Elza Erkip\IEEEauthorrefmark{1}, H. Vincent Poor\IEEEauthorrefmark{3}}\\
    \IEEEauthorblockA{\IEEEauthorrefmark{1}  NYU Tandon School of Engineering,  \{farideh,elza\}@nyu.edu}
    \IEEEauthorblockA{\IEEEauthorrefmark{2} Nokia - Bell Labs, Shirin.Jalali@alcatel-lucent.com}
    \IEEEauthorblockA{\IEEEauthorrefmark{3} Electrical Engineering  Department, Princeton University, poor@princeton.edu}
}

\maketitle

\begin{abstract}

The rate-distortion dimension (RDD) of an analog stationary process is studied as a  measure of complexity that captures the amount of information contained  in the process.  It is shown that the RDD of a process,  defined as two times the asymptotic ratio of its rate-distortion function $R(D)$ to $\log {1\over D}$ as the distortion $D$ approaches zero, is equal to its  information dimension (ID). This generalizes an earlier result by Kawabata and Dembo and provides  an operational approach to evaluate the ID of a process,  which previously was shown to be closely related to the effective dimension of the underlying process and also to the fundamental limits of compressed sensing. The relation between RDD and ID is illustrated for a piecewise constant process. 

\blfootnote{This work is part of a paper under review by  the {\em IEEE
Transactions on Information Theory,} available at  \cite{RJEP:16-arxiv}. This research
was supported by the National Science Foundation under Grant
CCF-1420575.}


\end{abstract}
\section{Introduction}

For discrete-alphabet signals, the Shannon entropy function $H(X)$ and the entropy rate  $\bar{H}(\Xbbf)=\lim_{n\to\infty}H(X_n|X^{n-1})$ measure the complexity of a random variable $X$ and  a stationary stochastic   process $\Xbbf=\{X_i\}$, respectively. Both of these measures are closely connected to the minimum number of bits per symbol required for representing stochastic sources \cite{cover} and can also be thought of as measures of signal structure. However, when we shift from discrete alphabet to continuous alphabet, both the entropy and the entropy rate become  infinite. Instead, for analog signals, the notion of information dimension (ID) introduced by R\'enyi \cite{Renyi:59} provides a framework that can be used to quantify signal structure. 


To illustrate what is meant for an analog process to be structured, consider a stationary memoryless (i.e., independent and identically distributed or i.i.d.) process $\Xbbf=\{X_i\}_{i=0}^{\infty}$ such that $X_i\sim (1-p)\delta_0+p f_c$, where $f_c$ denotes the probability density function (pdf) of an  absolutely continuous distribution and  $\delta_0$ denotes the Dirac measure with an atom at $0$. In other words, for each $i$, with probability $p\in[0,1]$, $X_i$ is exactly equal to zero; otherwise, it is drawn from $f_c$. By the strong law of large numbers, for large values of blocklength $n$, with probability approaching one,  a block $X^n$ generated by this source contains around $n(1-p)$ entries equal to zero, and the rest of the entries are real numbers in  the domain of $f_c$. To describe $X^n$ with a certain precision, for zero entries, it suffices to describe their locations. The number of bits required for this description does not depend on the reconstruction quality. However, for the remaining approximately $np$ elements of $X^n$, it  is known from rate-distortion theory that the required number of bits grows with the desired reconstruction quality. This intuitively suggests that $p$, which controls the number of non-zero elements in $X^n$, is a fundamental quantity related to the complexity and structure of $X^n$.  This intuition is accurately captured by the ID of this source which can be shown to be equal to $p$ \cite{Renyi:59}.  In fact, $\delta_0$ can be changed to any discrete probability distribution with finite entropy and the result will not change since the R\'enyi ID of a discrete source is 0.    

A further significance of the ID as a measure of structure is its relationship to the problem of compressed sensing. Consider the problem of recovering a signal $X_o^n$ from under-determined  measurements $Y^m=AX_o^n$, where $m<n$. It is known that if the input signal $X_o^n$ is sparse, or in general ``structured'', it can be accurately recovered from the  measurements, even if  $m$ is far fewer than $n$ \cite{Donoho:06,CandesT:06,RichModelbasedCS, ChRePaWi10, VeMaBl02, DoKaMe06}. For stationary memoryless processes, under some mild conditions on the distribution,   the R\'enyi ID of the first order marginal distribution of the source characterizes the fundamental limits of compressed sensing, \ie the minimum number of measurements required for asymptotically almost lossless recovery \cite{WuV:10}. The notion of the R\'enyi  ID is extended to  stationary processes in \cite{JalaliP:14-arxiv}, where it is proved that there is a direct relationship between the ID of a stationary process and the number of random linear measurements required for its universal recovery.

While the aforementioned results give an operational meaning to the ID of a signal, evaluating the ID of a stationary process is in general difficult. Kawabata and Dembo defined the rate-distortion dimension (RDD) of i.i.d. random variables (or vectors) based on the rate-distortion trade-off in the asymptoticly low distortion regime \cite{KawabataD:94}. They proved that for a random variable, its (upper and lower) RDD is equal to its (upper and lower) ID. 

The main contribution of this paper is to extend the notion of RDD to analog stationary processes, and to prove that, under some regularity conditions,  the  RDD of a stationary process is equal to its ID, defined in \cite{JalaliP:14-arxiv}.  This  provides an extension of the result of  Kawabata and Dembo to stochastic processes, and thereby  provides a  computationally feasible way of finding the ID of a stochastic process. In order to illustrate this, we compute the RDD of piecewise-constant stochastic processes, which are widely used to model many natural signals such as images. We derive upper and lower bounds on the rate-distortion functions of such signals,  and use these bounds to  evaluate the  RDD and hence, the ID of such processes. Furthermore, our results in  \cite{RJEP:16-arxiv} suggest that the RDD of a stochastic process is closely related to the fundamental limits of compressed sensing for the process, and hence RDD and ID can be thought of as measures of structure/complexity for arbitrary stationary stochastic processes.


The organization of the paper is as follows. Section \ref{sec:RDD} defines and examines the properties of ID and RDD. Section \ref{sec:RDD-ID} contains our main result which establishes  a connection between the ID  and the RDD of stochastic processes. Upper and lower bounds on the rate-distortion region of the piecewise constant source modeled by a first-order Markov process are provided in Section  \ref{sec:pwc_source} to illustrate the relationship between RDD and ID. Section \ref{sec:conc} concludes the paper.


\subsection{Notation}

Capital letters like $X$ and $Y$ represent random variables. 
For $x\in\mathds{R}$, $\lceil x \rceil$ ($\lfloor x \rfloor$) represents  the smallest (largest) integer larger (smaller) than $x$. For $b\in\mathds{N}^+$, $[x]_b$ denotes the $b$-bit approximation of $x$, \ie for $x=\lfloor x \rfloor+\sum_{i=1}^{\infty}(x)_i2^{-i}$, $(x)_i\in\{0,1\}$, $[x]_b=\lfloor x \rfloor+\sum_{i=1}^{b}(x)_i2^{-i}.$ Also, let $\langle x\rangle_b$ be defined as $\langle x\rangle_b={\lfloor bx\rfloor \over b}.$  For $x^n\in\mathds{R}^n$, $[x^n]_b$ and $\langle x^n\rangle_b$ are defined as $([x_1]_b,\ldots,[x_n]_b)$ and $(\langle x_1\rangle_b,\ldots,\langle x_n\rangle_b)$, respectively. 
Throughout the paper, $\log$ refers to the logarithm in base 2 . 

\section{Background}\label{sec:RDD}

In this section, we provide formal definitions of  ID and RDD and an overview of the literature.

\begin{definition}[R\'enyi information dimension \cite{Renyi:59}]\label{eq:def1}
The R\'enyi  upper and lower IDs of  an analog random variable $X$ are defined as
\[
\bar{d}(X)=\limsup_{b\to\infty} {H(\langle X\rangle_b)\over \log b},
\]
and
$
\underline{d}(X)=\liminf_{b\to\infty} {H(\langle X\rangle_b)\over \log b},
$
respectively. If the two limits coincide, the R\'enyi ID of $X$ is defined as $d(X)=\bar{d}(X)=\underline{d}(X)$.
\end{definition}
Definition \ref{eq:def1} can also be applied to analog vectors. For instance, for a random vector $X^n$, $\bar{d}(X^n)=\limsup_{b\to\infty} {H(\langle X^n\rangle_b)\over \log b}.$ 

 While the  above definition of the R\'enyi ID is in terms of the entropy of the $b$-level quantized version of $X$ normalized by the number of bits required for binary representation of it, $\log b$, as proved in Proposition 2 of \cite{WuV:10}, it can equivalently be defined in terms of the entropy of the $b$-bit quantized version of  $X$, $[X]_b$, normalized by $b$, \ie 
\[
\bar{d}(X)=\limsup_{b\to\infty} {H([X]_b)\over b},
\]
and $\underline{d}(X)=\liminf_{b\to\infty} {H([X]_b)\over b}.$


The notion of R\'enyi ID for random variables or vectors was extended in \cite{JalaliP:14-arxiv} to define the ID of analog stationary processes.

\begin{definition}[ID of a stationary process \cite{JalaliP:14-arxiv}]
  The  $k$-th order upper and lower IDs of   stationary process ${\Xbbf}=\{X_i\}_{i=-\infty}^{\infty}$ are defined as
   \[
  \bar{d}_k({\Xbbf})=\limsup_{b\to \infty} {1\over b}H([X_{k+1}]_b|[X^k]_b),
  \]
   and $ \underline{d}_k({\Xbbf})=\liminf_{b\to \infty} {1\over b}H([X_{k+1}]_b|[X^k]_b),$
respectively. The upper and lower ID of the process ${\Xbbf}$ are defined as
\[
\bar{d}_o({\Xbbf})=\lim_{k\to\infty}\bar{d}_k({\Xbbf})
 \]
 and $\underline{ d}_o({\Xbbf})=\lim_{k\to\infty}\underline{d}_k(\Xbbf),$
  respectively, when the limits exist. If $\bar{d}_o({\Xbbf})=\underline{d}_o({\Xbbf})$, the ID of process $\mathbb{\mathbf X}$, ${d}_o({\Xbbf})$, is defined as ${d}_o({\Xbbf})=\bar{d}_o({\Xbbf})=\underline{d}_o({\Xbbf})$.
\end{definition}
As proved in \cite{JalaliP:14-arxiv}, both $\bar{d}_k(\Xbbf)$ and $\underline{d}_k(\Xbbf)$ are both non-negative decreasing sequences in $k$. Hence, if they are also bounded, which is the case for instance for bounded sources,  their limits  as $k\to\infty$ also exist.

For a stationary memoryless process ${\Xbbf}=\{X_i\}_{i=-\infty}^{\infty}$,  this definition coincides  with that of R\'enyi\rq{}s ID of the first-order marginal distribution of the process $\Xbbf$.  That is $\bar{d}_o({\Xbbf})=\bar{d}(X_1)$ and $\underline{d}_o({\Xbbf})=\underline{d}(X_1)$. For sources with memory, taking the limit as   the memory parameter $k$ grows to infinity allows $d_o(\Xbbf)$ to capture the overall structure that is  present  in an analog stationary process. It can be proved that $d_o(\Xbbf)\leq 1$, for all bounded stationary processes, and if the stationary process $\Xbbf$ is structured,  $d_o(\Xbbf)$ is strictly smaller than one \cite{JalaliP:14-arxiv}.

Under some mild conditions on the distribution, \cite{WuV:10} proves that the R\'enyi ID of the first-order marginal distribution of a stationary memoryless process characterizes the fundamental limits of its compressed sensing. In other words, given a stationary memoryless process ${\Xbbf}$, asymptotically, as the blocklength $n$ grows to infinity,  the minimum number of linear projections ($m$) normalized by the blocklength ($n$) that is required for  recovering source $X^n$  is shown to be equal to $d(X_1)$. In \cite{JalaliP:14-arxiv}, it is shown that,  asymptotically, slightly more than $n\bar{d}_o({\Xbbf})$ random linear projections suffice  for \emph{universal} recovery of $X^n$ generated by any stationary process that satisfies some mixing conditions. These results provide an operational interpretation of the  ID of a random process.

The rate-distortion function of a stationary source measures the minimum number of bits per source symbol required for  achieving a given reconstruction quality. In some cases, as the reconstruction becomes finer, the behavior of  the rate-distortion function is connected to the level of structuredness of the source process and also to its ID mentioned earlier. In the rest of this section, we review the known results on these connections.

Consider a metric space  $(\mathds{R}^k,\rho)$, and random vector $X^k$. The  standard rate-distortion function \cite{cover} of vector $X^k$ under  distortion measure 
$d(x^k,\xh^k)=\rho(x^k,\xh^k)^r$, where $r>0$,  is defined as
\[
R_r(X^k,D)=\inf_{ \E[d(X^k,\Xh^k)]\leq D}I(X^k;\Xh^k).
\]
\begin{definition}[Rate-distortion dimension (RDD) of a random vector \cite{KawabataD:94}]
The upper and lower RDDs of $X^k$ are defined as
\[
\overline{\dim}_R(X^k)=r\limsup_{D\to0}{R_r(X^k,D)\over \log{1\over D}},
\]
and
$\underline{\dim}_R(X^k)=r\liminf_{D\to0}{R_r(X^k,D)\over \log{1\over D}},
$
 respectively. If $\overline{\dim}_R(X^k)=\underline{\dim}_R(X^k)$, the RDD of $X^n$ is defined as ${\dim}_R(X^k)=r\lim_{D\to0}{R_r(X^k,D)\over \log{1\over D}}$.
\end{definition}

The following theorem from \cite{KawabataD:94} establishes the connection between the R\'enyi ID of a random vector $X^k$ and its RDD, for a general distribution on $X^k$.

\begin{theorem}[Proposition 3.3 in \cite{KawabataD:94}] \label{thm:prop3-3}
Consider the metric space $(\mathds{R}^k,\rho)$, such that there exists $0<a_1\leq a_2<\infty$ for  which $a_1\max_{i=1}^k|x_i-\xh_i|\leq \rho(x^k,\xh^k)\leq a_2\max_{i=1}^k|x_i-\xh_i|,$
for all $x^k,\xh^k\in\mathds{R}^k$. Then,  for any distribution of $X^k$,
\[
\overline{\dim}_R(X^k)=\bar{d}(X^k),
\]
and
$\underline{\dim}_R(X^k)=\underline{d}(X^k),
$
where $\overline{\dim}_R(X^k)$,  and $\underline{\dim}_R(X^k)$ denote the upper and lower RDD of $X^k$ under fidelity constraint $d(x^k,\xh^k)=\rho(x^k,\xh^k)^r$.
\end{theorem}

\section{Equivalence of RDD and ID for Analog Processes} \label{sec:RDD-ID}

This section provides the main result of this paper which extends the notion of  RDD to stationary processes and establishes its connection of the ID of the process. 

Consider an analog stationary process ${\Xbbf}=\{X_i\}_{i=-\infty}^{\infty}$. The rate-distortion function $R({\Xbbf},D)$ of the source ${\Xbbf}$ under squared error distortion can be characterized  as \cite{book:Berger,Gallager}
\[
R({\Xbbf},D)=\lim_{m\to\infty} R^{(m)}({\Xbbf},D),
\]
where
\[
R^{(m)}({\Xbbf},D)=\inf_{\E[d_m(X^m,\Xh^m)]\leq D}{1\over m} I(X^m;\Xh^m)
\]
and
\begin{align}
d_m(x^m,\xh^m)={1\over m}\|x^m-\xh^m\|_2^2. \label{sq-err-distortion}
\end{align}

Note that with this distortion metric, we have $r=2$ and $R^{(m)}({\Xbbf},D)= {1\over m}R_2({X}^m,D)$. It can also be shown that $\inf_{m} R^{(m)}({\Xbbf},D)=R({\Xbbf},D)$ \cite{Gallager}.

\begin{definition}[RDD of a stationary process]
The upper and lower RDDs of a stationary process $\Xbbf$ is defined as
\[
\overline{\dim}_R({\Xbbf})=2\limsup_{D\to0}{R({\Xbbf},D)\over \log{1\over D}}
\]
and
$\underline{\dim}_R({\Xbbf})=2\liminf_{D\to0}{R({\Xbbf},D)\over \log{1\over D}}.$
If $\overline{\dim}_{R}(\Xbbf)=\underline{\dim}_{R}(\Xbbf)$, then the RDD of $\Xbbf$ is defined as $\dim_{R}(\Xbbf)=\overline{\dim}_{R}(\Xbbf)=\underline{\dim}_{R}(\Xbbf)$.
\end{definition}

The following theorem extends the equivalence of R\'enyi ID and RDD established in Theorem \ref{thm:prop3-3} for i.i.d. random vectors to stationary processes.

\begin{theorem}\label{thm:ID-eq-RDD}
For a stationary process $\Xbbf=\{X_i\}_{i=-\infty}^{\infty}$, assume that $\lim_{D\to 0}  { R^{(m)}({\Xbbf},D)\over \log {1\over D}}$ exists for all $m$. Then,
\[
{\dim}_{R}(\Xbbf) = \bar{d}_o({\Xbbf}).
\]
\end{theorem}

The main ingredients of the proof of  Theorem \ref{thm:ID-eq-RDD} are the following two lemmas.


\begin{lemma}\label{lemma:connect-UID-URDD}
For any stationary process $\Xbbf$, we have
\[
\overline{\dim}_{R}(\Xbbf) \leq \bar{d}_o({\Xbbf})\leq \inf_{m} 2\Big(\limsup_{D\to0}  { R^{(m)}({\Xbbf},D)\over \log {1\over D}}\Big).
\]
\end{lemma}

\begin{lemma}\label{lemma:uniform_conv}
Assume that $\lim_{D\to 0}  { R^{(m)}({\Xbbf},D)\over \log {1\over D}}$ exists for all $m$, and also there exists $\sigma_{\max}^2>0$, such that ${R^{(m)}({\Xbbf},D)}$ uniformly converges  to ${R({\Xbbf},D)}$, for $D\in(0,\sigma_{\max}^2)$, as $m$ grows to infinity. Then, ${\dim}_{R}(\Xbbf) = \bar{d}_o({\Xbbf}).$

\end{lemma}


\begin{proof}[Proof of Lemma \ref{lemma:connect-UID-URDD}]
Given $k$, define distance measure $\rho_k$ such that for $x^k,\xh^k\in\mathds{R}^k$, $\rho_k(x^k,\xh^k)\triangleq \sqrt{kd_k(x^k,\xh^k)}$ where $d_k(\cdot,\cdot)$ is defined in (\ref{sq-err-distortion}). Note that $(\mathds{R}^k,\rho_k)$ is a metric space. Furthermore, since $\max_{i=1}^k|x_i-\xh_i|\leq \rho_k(x^k,\xh^k)\leq \sqrt{k} \max_{i=1}^k|x_i-\xh_i|$, from Theorem \ref{thm:prop3-3},
\[
2\limsup_{D\to0}{kR^{(k)}({\Xbbf},{D\over k})\over \log{1\over D}}=\bar{d}(X^k).
\]
By a change of variable, $2\limsup_{D\to0}{kR^{(k)}({\Xbbf},D)\over \log{1\over D}+\log{1\over k}}=\bar{d}(X^k),$
or
\[
2\limsup_{D\to0}{R^{(k)}({\Xbbf},D)\over \log{1\over D}}={1\over k}\bar{d}(X^k).
\]
Taking the limit of both sides as $k$ grows to infinity, and employing Lemma 2 from \cite{JalaliP:14-arxiv}, which shows that the upper ID of a process $\Xbbf$ can alternatively be represented as
\[
\bar{d}_o({\Xbbf})=\lim_{k\to\infty}{1\over k}\bigg(\limsup_{b\to\infty} {H([X^k]_b)\over b}\bigg),
\]
yields
\begin{align}
\lim_{k\to\infty}\bigg(2\limsup_{D\to0}{R^{(k)}({\Xbbf},D)\over \log{1\over D}}\bigg)&=\lim_{k\to\infty}{1\over k}\bar{d}(X^k)\nonumber\\
&=\bar{d}_o({\Xbbf}).\label{eq:application-lemma-1}
\end{align}

Since $R^{(k)}({\Xbbf},D)\geq \inf_mR^{(m)}({\Xbbf},D)$, from \eqref{eq:application-lemma-1},
\begin{align*}
\bar{d}_o({\Xbbf})&\geq \lim_{k\to\infty}\bigg(2\limsup_{D\to0}{\inf_m R^{(m)}({\Xbbf},D)\over \log{1\over D}}\bigg)\nonumber\\
&\stackrel{(a)}{=} \lim_{k\to\infty}\bigg(2\limsup_{D\to0}{R({\Xbbf},D)\over \log{1\over D}}\bigg)=\overline{\dim}_{R}(\Xbbf),
\end{align*}
where (a) follows from the fact that $R(\Xbbf,D)=\inf_{m} R^{(m)}(\Xbbf,D)$ \cite{Gallager}. This proves the lower bound in the desired result.

To prove the upper bound, fix a positive integer $m\in\mathds{N}$. Any integer $k$ can be written as $k=sm+r$, where $r\in\{0,\ldots,m-1\}$. Since $kR^{(k)}({\Xbbf},D)$ is  a sub-additive sequence \cite{Gallager}, and $k=m+\ldots+m+r$, $kR^{(k)}({\Xbbf},D)\leq sm R^{(m)}({\Xbbf},D)+rR^{(r)}({\Xbbf},D),$ it follows that
or
\begin{align}
R^{(k)}({\Xbbf},D)\leq {sm\over k} R^{(m)}({\Xbbf},D)+{r\over k}R^{(r)}({\Xbbf},D).\label{eq:sub-additive}
\end{align}
Combining \eqref{eq:application-lemma-1} and \eqref{eq:sub-additive}, it follows that
\begin{align}
\bar{d}_o({\Xbbf})\leq &\; 2\lim_{k\to\infty}\bigg(\limsup_{D\to0} {sm\over k} { R^{(m)}({\Xbbf},D)\over \log {1\over D}}\bigg)\nonumber\\
&+2\lim_{k\to\infty}\bigg( \limsup_{D\to0}{r\over k}{R^{(r)}({\Xbbf},D)\over \log{1\over D}}\bigg)\nonumber\\
= &\; 2\lim_{k\to\infty}\Big({sm\over k}\Big)\bigg(\limsup_{D\to0}  { R^{(m)}({\Xbbf},D)\over \log {1\over D}}\bigg)\nonumber\\
&+2\lim_{k\to\infty}\Big({r\over k}\Big)\bigg( \limsup_{D\to0}{R^{(r)}({\Xbbf},D)\over \log{1\over D}}\bigg)\nonumber\\
= &\; 2\bigg(\limsup_{D\to0}  { R^{(m)}({\Xbbf},D)\over \log {1\over D}}\bigg).\label{eq:ub_wo_inf}
\end{align}
Since $m$ is selected arbitrarily, we can take the infimum of the  right hand side of \eqref{eq:ub_wo_inf} and derive the desired result.
\end{proof}


\begin{proof}[Proof of Lemma \ref{lemma:uniform_conv}]
By the lemma's assumption, $\overline{\dim}_{R}(\Xbbf) ={\dim}_{R}(\Xbbf) $; therefore, from Lemma \ref{lemma:connect-UID-URDD}, \begin{align}\label{eq:b1}
{\dim}_{R}(\Xbbf) \leq \bar{d}_o({\Xbbf})\leq 2\Big(\lim_{D\to0}  { R^{(m)}({\Xbbf},D)\over \log {1\over D}}\Big),
\end{align}
for all $m$. Given the uniform convergence assumption, for any $\e>0$, there exists $m_{\e}\in\mathds{N}$, such that for all $m>m_{\e}$,
\begin{align}\label{eq:cond1}
\left|{R^{(m)}({\Xbbf},D)\over \log{1\over D}}-{R({\Xbbf},D)\over \log{1\over D}}\right|<\epsilon,
\end{align}
for all $D\in(0,\sigma^2_{\max})$.

On the other hand, for any $\e'>0$ and  $m$, there exists $\d_{\e',m}>0$, such that for all $D\in(0,\d_{\e',m})$,
\begin{align}\label{eq:cond2}
\lim_{D\to0}  { R^{(m)}({\Xbbf},D)\over \log {1\over D}}\leq { R^{(m)}({\Xbbf},D)\over \log {1\over D}} +\e'.
\end{align}
Also, for any $\e''>0$, there exists $\d_{\e''}>0$, such that for all $D\in(0,\d_{\e''})$,
\begin{align}\label{eq:cond3}
 { R({\Xbbf},D)\over \log {1\over D}} \leq \frac{1}{2}\left({\dim}_{R}(\Xbbf) +\e''\right).
\end{align}

Therefore, for any $\e,\e'$ and $\e''$, choosing $m>m_{\e}$, and $D\in(0,\min(\d_{\e',m},\d_{\e''}))$, and combining \eqref{eq:cond1}, \eqref{eq:cond2} and \eqref{eq:cond3} yields
\begin{align}\label{eq:b2}
\bar{d}_o({\Xbbf})\leq  {\dim}_{R}(\Xbbf) +\e+\e'+\e''.
\end{align}
Since $\e,\e'$ and $\e''$ are selected arbitrarily, combining \eqref{eq:b1} and \eqref{eq:b2} proves that ${\dim}_{R}(\Xbbf) = \bar{d}_o({\Xbbf})$.
\end{proof}


\begin{proof}[Proof of Theorem \ref{thm:ID-eq-RDD}]
It is shown in \cite{WynerZ:71} that for any stationary process $\Xbbf$
\begin{align}
|R^{(m)}(\Xbbf,D)-R(\Xbbf,D)|\leq{1\over m}I(X^m;X^0_{-\infty}). \label{eq:bd-Wyner-Ziv}
\end{align}
Note that while some of the results in \cite{WynerZ:71} hold only for sources that are either absolutely continuous or discrete, as shown in \cite{RJEP:16-arxiv}, this bound holds for  sources with general distributions. Since the right hand side of \eqref{eq:bd-Wyner-Ziv} does not depend on $D$, it shows that $R^{(m)}(\Xbbf,D)$ uniformly converges to $R(\Xbbf,D)$ for all $D>0$. On the other hand, for any $0<\sigma_{\max}<1$, and any $D\in(0,\sigma_{\max}^2)$, $0<1/\log{1\over D}<1/\log{1\over \sigma_{\max}^2}$. Therefore, ${R^{(m)}(\Xbbf,D)\over \log{1\over D}}$ uniformly converges  to ${R(\Xbbf,D)\over \log{1\over D}}$, for $D\in(0,\sigma_{\max}^2)$, and by Lemma \ref{lemma:uniform_conv}, ${\dim}_{R}(\Xbbf) = \bar{d}_o({\Xbbf}).$
\end{proof}

For an i.i.d. source $\Xbbf$, under some mild conditions, ${d}_o(\Xbbf)$ characterizes the fundamental limits of compressed sensing \cite{WuV:10}. In other words, asymptotically,  almost lossless recovery of $X^n$ generated by the source  $\Xbbf$ from measurements $Y^m=AX^n$ is feasible, if and only if  the normalized number of  measurements ($m/n$) is larger than $d_o(\Xbbf)$. If the rate-distortion function of the source satisfies the condition of  Theorem \ref{thm:ID-eq-RDD}, then  ${\dim}_R({\Xbbf})=\bar{d}_o({\Xbbf})$, which implies  that the RDD of an i.i.d.~process can also be used to characterize its compressed sensing fundamental limits.  
On the other hand, compression-based compressed sensing of stochastic processes is  studied in \cite{RJEP:16-arxiv}.  It is shown in  \cite{RJEP:16-arxiv} that there exists a compression-based recovery algorithm that  achieves almost lossless recovery by using slightly more than $n\overline{\dim}_R({\Xbbf})$ random linear measurements. This implies   that  $\overline{\dim}_{R}(\Xbbf)$ is achievable for  general sources.  (Note that, by Lemma \ref{lemma:connect-UID-URDD},    in general $\overline{\dim}_{R}(\Xbbf)\leq\bar{d}_o({\Xbbf})$.)

\begin{remark}
\label{rem:id}
Theorem \ref{thm:ID-eq-RDD}, by proving  the equivalence of ID and RDD, provides  a potentially easier  path to computing the ID of stochastic  processes.  Note that while to directly compute the ID of a process one needs to take the limit over the quantized approximations and then over the memory length,  to be able to calculate the RDD of a process, the exact characterization of the rate-distortion function is not required. In fact, it is easy to see that it would be enough to have upper and lower bounds on  the rate-distortion function of the source, $R(\Xbbf,D)$, that are within a reasonable gap. More precisely,  as long as the gap between the bounds grows as $o(\log{1 \over D})$, they can be used to evaluate the RDD.  Moreover, since the RDD  depends only on the low-distortion behavior  of the rate-distortion function,  studying its asymptotic small distortion performance  is sufficient for computing the RDD, and  as by Theorem \ref{thm:ID-eq-RDD}, ID of a source, without knowing the rate-distortion function explicitly.  For instance,  \cite{Gyorgy99} studies the asymptotic behavior of the rate-distortion function of  some stochastic sources and employs those results to evaluate the RDD of some i.i.d.~processes.
\end{remark}

The next section illustrates computation of RDD and its relation to ID  for a piecewise constant process. 


\section{RDD of a Piecewise-Constant Process} \label{sec:pwc_source}

In general, deriving the rate-distortion function of sources with memory is  very challenging. For instance, even for the binary symmetric Markov chain, the rate-distortion function is not known, except in a low-distortion region \cite{Gray:71}, and we have to resort to upper and lower bounds  \cite{Berger:77,JalaliW:07}. 

In this section we consider a piecewise constant signal modeled by a first order Markov process ${\Xbbf}=\{\X_\indx\}_{\indx=1}^\infty$, such that  conditioned on $\X_{\indx-1}=\x_{\indx-1}$, $\X_\indx$ is distributed according to $(1-p)\delta_{\x_{\indx-1}}+pf_c$, where $f_c$ denotes the pdf of an absolutely continuous distribution with bounded support, defined over an interval $(l,u)$. In other words, at each time $i$,  the process either makes a jump and takes a value drawn from distribution $f_c$, or it stays at $X_{i-1}$. The decision is made based on the outcome of an i.i.d. $\operatorname{Bern} \left({p}\right)$ random process independent of all past values of $\Xbbf$. While the output of this source is  not  sparse,  it is clearly a structured process.  The following theorem provides upper and lower bounds on  $R(\mathbf{X},D)$ of the piecewise-constant source. While there is a gap between the bounds on  $R(\mathbf{X},D)$, since the gap does not depend on $D$, as shown in the following corollary, they can be used to evaluate  RDD of the source exactly.

\begin{theorem}\label{thm:r-d-bound}
Consider a first-order stationary Markov process ${\Xbbf}=\{\X_\indx\}_{\indx=0}^\infty$, such that  conditioned on $\X_{\indx-1}=\x_{\indx-1}$, $\X_\indx$ is distributed according to $(1-p)\delta_{\x_{\indx-1}}+pf_c$, where $f_c$ denotes the pdf of an absolutely continuous distribution with bounded support,  $(l,u)$. 
 If $d_{\max}\triangleq \sup_{x,\xh\in(l,u)}d(x,\xh)<\infty$, then
\[
pR_{f_c}(D)\leq R(\mathbf{X},D)\leq H(p)+pR_{f_c}(D),
\]
  where $R_{f_c}(D)$  and $H(p)$  denote the rate distortion function of an i.i.d.~process distributed according to pdf $f_c$, and the binary entropy function ($-p\log_2 p-(1-p)\log_2(1-p)$), respectively.
\end{theorem}

\begin{proof}
To prove the upper bound (achievability), we consider a code that describes  the positions of the jumps  losslessly at rate $H(p)$. Since the source is piecewise constant, after describing the positions of the jumps,  the encoder removes the repeated values and applies a lossy compression code of blocklength  close to $np$. Therefore, to describe the values at distortion $D$ the encoder roughly needs to spend $npR_{f_c}(D)$ bits.  For the lower bound (converse), we consider a genie-aided  decoder that has access to the positions of the jumps.  Then intuitively, to describe the values at distortion $D$, it still needs a rate of at least $pR_{f_c}(D)$.  The proof presented in \cite{RJEP:16-arxiv} makes these steps formal by properly analyzing the reduced block length which is a random number.
\end{proof}

\begin{corollary}\label{thm:RDD-piecewise-constant}
For the piecewise constant source in Theorem \ref{thm:r-d-bound}, we have
 $
 {\dim}_R({\Xbbf})=\bar{d}_o(\Xv)=p.  
 $ In other words, the RDD is equal to $p$ which is in turn equal to the ID of this source.
 \end{corollary}
 \begin{proof}
 Given the bound on the rate-distortion process derived in Theorem \ref{thm:r-d-bound}, it is easy to directly derive the  RDD of such a source. More precisely, given the upper bound, it follows that
 \begin{align*}
 \overline{\dim}_R({\Xbbf})=2\limsup_{D\to0}{R({\Xbbf},D)\over \log{1\over D}}
 = p(\limsup_{D\to0} {R_{f_c}(D) \over \log{1\over D}})
 =p,
 \end{align*}
 where the last step follows from \cite{Renyi:59} and \cite{KawabataD:94}. Similarly, given the lower bound, we have $\underline{\dim}_R({\Xbbf})\geq p$. Therefore, $p\leq\underline{\dim}_R({\Xbbf})\leq\overline{\dim}_R({\Xbbf})\leq p$. In other words, for this source RDD exists and is equal to ${\dim}_R({\Xbbf})=p$. Hence, the condition of  Theorem \ref{thm:ID-eq-RDD} holds and we have $ {\dim}_R({\Xbbf})=\bar{d}_o(\Xv).$
 This agrees with the ID of this source found in Theorem 2 in \cite{JalaliP:14-arxiv}, $\bar{d}_o(\Xv)=\underline{d}_o(\Xv)=p.$
\end{proof}

Corollary \ref{thm:RDD-piecewise-constant} states that the RDD of the piecewise constant source described in Theorem \ref{thm:r-d-bound} is equal to $p$, which is also the ID of this process \cite{JalaliP:14-arxiv}. While \cite{JalaliP:14-arxiv} directly computes the ID of such processes, Theorem \ref{thm:ID-eq-RDD} provides an easier alternate method for computing the ID as suggested in Remark~\ref{rem:id}. 




\section{Conclusions}\label{sec:conc}

In this paper we have defined the RDD of stationary  processes, as a generalization of the RDD of stochastic vectors introduced in \cite{KawabataD:94}. We have proved that under some mild conditions the RDD of a stationary process is equal to its ID introduced in \cite{JalaliP:14-arxiv}.
This gives an operational method to evaluate the ID of a stationary process, which was previously    shown to be related to the fundamental limits of compressed sensing \cite{WuV:10,JalaliP:14-arxiv,RJEP:16-arxiv}.


\bibliographystyle{unsrt}
\bibliography{myrefs}

\setcounter{equation}{0}
\renewcommand{\theequation}{\thesection.\arabic{equation}}

\end{document}